\documentclass{article}

\usepackage[english]{babel}
\usepackage[latin1]{inputenc}
\usepackage{amssymb}
\usepackage{amsmath}
\usepackage{amsthm}
\usepackage{enumerate}
\usepackage{url}
\usepackage{graphicx}
\usepackage{pst-node}

\usepackage{a4}
\usepackage{tikz}

\usepackage{amssymb,amsmath,amsthm}
\usepackage{enumerate}

\newtheorem{theorem}{theorem}
\newtheorem{definition}[theorem]{Definition}

\newtheorem{observation}[theorem]{Observation}

\newtheorem{defn}[theorem]{Definition}

\newtheorem{lem}[theorem]{Lemma}

\theoremstyle{definition}

\newcommand{\Complex}{\mathbb{C}}

\newcommand{\X}{\mathbf{X}} 
\newcommand{\Y}{\mathbf{Y}} 

\newcommand{\B}{$\textsf{B }$}
\newcommand{\C}{$\textsf{C }$}
\newcommand{\K}{$\textsf{K }$}
\newcommand{\I}{$\textsf{I }$}

\newcommand{\For}{{ \textsf{F}^k}}
\newcommand{\F}{\textsf{F}}             \newcommand{\NnF}{\emph{\texttt{F}}}
\newcommand{\CL}{\textsf{CL} }          \newcommand{\NnCL}{\emph{\texttt{CL}} }
\newcommand{\INT}{\textsf{INT} }        
\newcommand{\G}{\textsf{G} }            \newcommand{\NnG}{\emph{\texttt{G}} }
\newcommand{\BCK}{\textsf{BCK} }        \newcommand{\NnBCK}{\emph{\texttt{BCK}} }
\newcommand{\BCI}{\textsf{BCI} }        
\newcommand{\Int}{\textsf{INT} }        \newcommand{\NnInt}{\emph{\texttt{INT}} }
\newcommand{\SN}{\textsf{SN} }          \newcommand{\NnSN}{\emph{\texttt{SN}} }
\newcommand{\LN}{\textsf{LN} }          \newcommand{\NnLN}{\emph{\texttt{LN}} }
\newcommand{\PEIRCE}{\textsf{PEIRCE} }  \newcommand{\NnPEIRCE}{\emph{\texttt{PEIRCE}} }
\newcommand{\EVEN}{\textsf{EVEN} }      \newcommand{\NnEVEN}{\emph{\texttt{EVEN}} }

\newcommand{\ar}{\hspace*{0.1mm} \Rightarrow }
\newcommand{\impl}{\hspace*{0.1mm} \Rightarrow }

\newcommand{\norm}[1]{\left\Vert #1 \right\Vert}

\newcommand{\liminfty}[1]{\lim_{k \rightarrow \infty} #1}
\newcommand{\limkinfty}[1]{\lim_{k \rightarrow \infty} #1}
\newcommand{\limsupinfty}[1]{\limsup_{n \rightarrow \infty} #1}
\newcommand{\liminfinfty}[1]{\liminf_{n \rightarrow \infty} #1}

\newcommand{\bt}[2]{\begin{#1}\label{#1:#2}}
\newcommand{\et}[1]{\end{#1} }

\newcommand{\Proof}{{\it \noindent Proof. }}
\def\boksik{\hspace*{1mm}\hfill $\Box$}

 \newcommand{\abs}[1]{\left\vert#1\right\vert}
 \newcommand{\set}[1]{\left\{#1\right\}}

\newcommand{\comment}[1]{   }

\def\Newton#1#2{
   \def\ss{\arraystretch} \def\arraystretch{0.8}
   \left(\!\!\begin{array}{c} #1\\ #2\end{array}\!\!\right)
   \def\arraystretch{\ss} }

\newcommand{\subterm}[4]{%
  \fill[draw=gray,fill=gray!20!white,very thick] #1 +(0,#3) --%
  +(#2,-#3) -- +(-#2,-#3) -- cycle; \draw #1+(0,-0.2) node
  [circle] {#4}; }

\newcommand\lmbda[4]{%
  \fill[black] #1 circle (3pt);
  \draw #1 node [circle, anchor=#4] {#3} -- #2;
}

\newcommand\app[3]{%
  \fill[black] #1 circle (3pt);
  \draw #1 -- #2;
  \draw #1 -- #3;
}

\begin{document}
\title{How big is \BCI  fragment of \BCK    logic\thanks{This work was partially supported
by grant number N206 3761 37, Polish Ministry of Science and Higher Education.}}

\author{Katarzyna Grygiel, Pawe{\l} M. Idziak and Marek Zaionc\thanks{Faculty of Mathematics and Computer Science, Theoretical Computer Science, Jagiellonian University, {\L}ojasiewicza 6, 30-348 Krak\'ow, Poland.
Email: {\tt \{grygiel,idziak,zaionc\}@tcs.uj.edu.pl}}}

\date{\today}
\maketitle

 \begin{abstract} We investigate quantitative properties of \BCI and \BCK logics.
The first part of the paper compares the number of formulas provable in \BCI versus \BCK logics.
We  consider formulas built on implication and  a fixed set of $k$ variables.
We investigate the proportion between the number of such formulas of a given length $n$
provable in \BCI logic against the number of formulas of length $n$ provable in richer  \BCK logic. We examine an asymptotic behavior of this fraction when length $n$ of formulas tends to infinity. This limit gives a probability measure that randomly chosen \BCK formula is also provable in \BCI. We prove that this probability tends to zero as the number
of variables tends to infinity. The second part of the paper is devoted to the number of lambda terms representing proofs of \BCI and \BCK logics. We build a proportion between number of such proofs of the same length $n$ and we investigate asymptotic behavior of this proportion when length 
of proofs tends to infinity. We demonstrate that with probability $0$ a randomly chosen \BCK proof is also a proof of a \BCI formula.
\end{abstract}

{\bf Keywords:} \BCK and \BCI logics, asymptotic probability in logic, analytic combinatorics.

\section{Introduction}

The results presented in this paper are a part of research in which the likelihood
of truth is estimated for various propositional logics with a limited number of variables.
Probabilistic methods appear to be very
 powerful in combinatorics and computer science. From a point of view of
 these methods we investigate a typical object chosen
 from some set. For formulas in the fixed propositional language, we investigate the
proportion between the number of valid formulas of a given length
$n$ against the number of all formulas of length $n$.
Our interest lies  in finding the limit of that
fraction when $n$ tends to infinity.
If the limit exists, then it is
represented by a real number which we may call {\it the density} of the investigated logic.
In general, we are also interested
in finding  the {\it `density'} of some other classes of formulas.
Good presentation and overview of asymptotic methods for random
boolean expressions can be found in the paper \cite{gardy-dmtcs} of Gardy.
For the purely implicational logic of one variable (and at the same time
simple type systems),  the exact value of the density of true formulas
was computed by Moczurad, Tyszkiewicz and Zaionc in~\cite{mtz00}.
The classical logic of one
variable and the two connectives of implication and negation was studied in
Zaionc~\cite{zaionc05}; over the same language, the exact proportion between
intuitionistic and classical logics was determined by Kostrzycka and
Zaionc in~\cite{kos-zaionc03}.
Asymptotic  identity between classical and intuitionistic logic of implication has
been proved  in  Fournier, Gardy, Genitrini and Zaionc in~\cite{FGGZ07}.
Some variants involving expressions with other logical
connectives have also been considered.
Genitrini and Kozik have studied the influence of adding the connectors $\lor$ and $\land$
to implication in~\cite{GK}, while
Matecki in~\cite{mat05} considered the case of the single equivalence connector.
For two connectives again, the \emph{ and/or} case has already received much
attention -- see Lefmann and Savick\'{y}~\cite{LS97}, Chauvin, Flajolet, Gardy and
Gittenberger~\cite{CFGG04}, Gardy and Woods~\cite{GW05}, Woods~\cite{w05} and Kozik~\cite{kozik08}.
Let us also mention the survey  \cite{gardy-dmtcs} of Gardy on the probability distributions on Boolean functions induced by random Boolean expressions; this survey deals with the whole set of Boolean functions on some finite number of variables.

\section{\BCK and \BCI logics}

The logics \BCK and \BCI are ones of several pure implication calculi.
Its name comes from the connection with the combinators \B, \C, \K  and
\I (see \cite{Curry_Hindley_Seldin}).
From the perspective of type theory, \BCK and \BCI can be viewed as the set of types of a certain restricted family of lambda terms, via the Curry-Howard isomorphism.
Formally, logics \BCK and \BCI can be defined, each one separately, as Hilbert systems by three axiom schemes and detachment rule. Namely \BCK is based on \B, \C and \K while  \BCI is based on
\B, \C and \I where:

\begin{itemize}
 \item[ (\B)] $(\varphi \impl \psi ) \impl ((\chi \impl \varphi ) \impl ( \chi \impl \psi )$ (prefixing)
 \item[ (\C)] $( \varphi \impl (\psi  \impl \chi)) \impl (\psi \impl (\varphi \impl \chi))$ (commutation)
 \item[ (\K)] $\varphi \impl (\psi \impl \varphi) $
 \item[(\I  )  ] $\varphi \impl \varphi$ (identity)
 \end{itemize}

In \BCK we are able to prove \I therefore
the logic \BCI is a subset of the \BCK.
Let us observe that implicational formulas may be seen as rooted binary trees.

\begin{definition}\label{formula_tree}
By a formula tree we mean the rooted binary tree in which nodes are labeled by $\impl$ and have two successors left and right while leaves of the tree are labeled by variables.
\end{definition}

\begin{definition}\label{formula_tree1}
With every implicational formula  $\varphi$ we associate the formula  tree $G(\varphi)$ in the following way:
\begin{itemize}
  \item If $x$ is a variable, then $G(x)$ is a single node labeled with $x$.
  \item Tree $G( \varphi \impl \psi ) $ is the  tree with the new root labeled with $\impl$ and two subtrees: left $G(\varphi)$ and right $G(\psi)$.
\end{itemize}
\end{definition}

\section{$\lambda$-calculus as a proof system}\label{lambda}

Lambda calculus is a standard mechanism for proof system representation for various propositional calculi. By the Curry-Howard isomorphism there is a one-to-one correspondence between provable formulas in intuitionistic implicational logic and types of
closed lambda calculus terms. Moreover, proofs of formulas correspond to
typable terms. We start with presenting some fundamental concepts of the $\lambda$-calculus, as well as with some new definitions used in this paper.

\begin{definition}
Let $V$ be a countable set of variables. The set $\Lambda$ of $\lambda$-terms is defined by the following grammar:
\begin{enumerate}
\item every variable is a lambda term,
\item if $t$ and $s$ are lambda terms then $ts$ is a lambda term,
\item if $t$ is a  lambda terms and $x$ is a variable then $\lambda x .t$ is a lambda term.
\end{enumerate}
\end{definition}

As usual, $\lambda$-terms are considered modulo the $\alpha$-equivalence, i.e. two terms which differ only by the names of bounded variables are considered equal.
Observe that $\lambda$-terms can be seen as rooted unary-binary trees.

\begin{definition}
By a lambda tree we mean the following rooted graph with two kinds of edges: undirected and directed.
One distinguished node is called the root of the graph.
The graph induced by undirected edges is a rooted tree with the distinguished node being the root of it.
There are two kinds of internal nodes labeled by $@$ and by $\lambda$. Nodes labeled by $@$ have two successors left and right. Nodes labeled with $\lambda$ have only one successor. Leaves of the tree are either labeled by variables or are connected by directed edge with the one of $\lambda$ nodes placed on the path from it to the root.

\end{definition}

\begin{definition}
With every lambda term $t$ we associate the lambda tree $G(t)$ in the following way:
\begin{itemize}
  \item If $x$ is a variable then $G(x)$ is a single node labeled with $x$.
  \item Lambda tree $G( P Q) $ is a lambda tree with the new root labeled with $@$ and connected by two new undirected edges with roots of two lambda subtrees left $G(P)$ and right $G(Q)$.
  \item Tree $G(\lambda x. P)$ is obtained from $G(P)$ in four steps:

   \begin{itemize}
     \item Add new root node labeled with $\lambda$.
     \item Connect new root by undirected edge with the root of $G(P)$.
     \item Connect all leaves of $G(P)$ labeled with $x$ by directed edges with the new root.
     \item Remove all labels $x$ from $G(P)$.
   \end{itemize}

\end{itemize}
\end{definition}

\bigskip
\bigskip

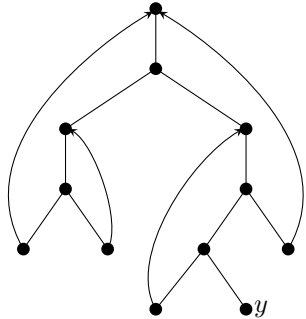
\begin{figure}[h]
  \centering
  \begin{tikzpicture}[>=stealth,scale=.8]
    \lmbda{(0,5)}{(0,4)}{}{east}
    \app{(0,4)}{(-1.5,3)}{(1.5,3)}
    \lmbda{(-1.5,3)}{(-1.5,2)}{}{east}
    \app{(-1.5,2)}{(-2.2,1)}{(-0.8,1)}
    \lmbda{(1.5,3)}{(1.5,2)}{}{east}
    \app{(1.5,2)}{(0.8,1)}{(2.2,1)}
    \app{(0.8,1)}{(0,0)}{(1.5,0)}
    \fill[black] (-2.2,1) circle (3pt)
                 (-0.8,1) circle (3pt)
                 (0,0) circle (3pt)
                 (1.5,0) circle (3pt)
                 (2.2,1) circle (3pt);
   \draw [->] (-2.2,1) .. controls +(120:2cm) and +(-150:1cm) .. (-0.05,4.95);
   \draw [->] (-0.8,1) .. controls +(60:0.5cm) and +(-30:0.5cm) .. (-1.45,3);
   \draw [->] (0,0) .. controls +(120:1cm) and +(-150:1cm) .. (1.45,3);
   \draw [->] (2.2,1) .. controls +(60:2cm) and +(-30:1cm) .. (0.05,4.95);
   \draw (1.75,0) node {$y$};
  \end{tikzpicture}
\caption{\rm The lambda tree representing the term $\lambda z.(\lambda u. zu)((\lambda u. uy)z)$}
\label{fig:headlambdaterm}
\end{figure}


\begin{observation}
 If $T$ is a lambda tree, then $T=G(M)$ for some lambda term $M$.
 Terms  $M$ and $N$ are $\alpha$-equivalent iff  $G(M) = G(N)$. Free variables of term $M$ are the same as variables labeling leafs of $G(M)$.
 \end{observation}

We often use (without giving the precise definition) the classical terminology about trees (e.g. path, root, leaf, etc.). A path from the root to a leaf is called a branch.

\comment{
\begin{definition}
The term of the form $(\lambda x.P)Q$ is called  $\beta$ redex.
A lambda term is in normal form if it does not contain $\beta$ redex subterm.
The least relation $\triangleright$ on terms satisfying $ (\lambda x.P)Q \triangleright P[x:=Q]$ is called $\beta$ reduction.
A term $M$ is normalizing if there is reduction sequence starting from $M$ and ending in a normal form $N$. A term $M$ is strongly normalizing if all reduction sequences are finite.  By $SN$ we mean all terms which are strongly normalizing.
\end{definition}

As an example we can see lambda tree representation of redex which is  a subterm of some lambda tree.
Therefore $\beta$ conversion can be seen as a operation on lambda trees.

\begin{figure}[h]
  \centering
  \begin{tikzpicture}[>=stealth,scale=.8]
    \draw[dotted, very thick] (-3,3) -- (-3,2);
    \app{(-3,2)}{(-4,1.5)}{(-2,1.5)}
    \lmbda{(-4,1.5)}{(-4,0.5)} {$\lambda$}{east}
    \subterm{(-4,0)}{1}{.5}{$t$}
    \subterm{(-2,1)}{.5}{.5}{$u$}
    \draw [->] (-4.5,-0.5) .. controls +(-150:2cm) and +(-120:1cm) .. (-4.05,1.45);
    \draw [->] (-3.5,-0.5) .. controls +(-30:2cm) and +(-60:1cm) .. (-3.95,1.45);
  \draw[dotted, very thick] (2,3) -- (2,2);
  \subterm{(2,1.5)}{1}{.5}{$t$}
  \subterm{(1.3,0.5)}{.5}{.5}{$u$}
  \subterm{(2.7,0.5)}{.5}{.5}{$u$}
  \end{tikzpicture}
\caption{\rm $\beta$-reduction scheme}
\label{fig:headlambdaterm}
\end{figure}
}

\subsection{\BCI and \BCK classes of lambda terms}

The Curry-Howard isomorphism for proof representation in
entire  intuitionistic implicational logic can be restricted to weaker logics. Therefore, we can
look at the axioms (\B), (\C), (\I) and (\K) of \BCI and \BCK logics as at types \`{a}
la Curry in the typed lambda calculus. Lambda terms for which these types
are principal (the most general ones) are respectively (\cite{hindley}):

\begin{align*}
{\bf B} & \equiv \lambda x y z. x (y z),\\
{\bf C} & \equiv \lambda x y z. x z y,\\
{\bf I} & \equiv \lambda x . x,\\
{\bf K} & \equiv \lambda x y . x.
\end{align*}

\bt{thm}{I}
The combinator \I is provable in \BCK.
\et{thm}

\Proof
For example the lambda term $(\textbf{C} \textbf{K}) \textbf{K}  = \lambda z.z$ forms proof for the combinator \I in logic \BCK. \boksik

\bt{thm}{subset}
\BCI is a proper subset of \BCK
\et{thm}

\Proof Inclusion follows from Theorem \ref{thm:I}. Combinator \K is not provable in logic \BCI (see
\cite{blok_pigozzi89}, for example). \boksik

\bigskip
We are going to isolate the special set of \BCK provable formulas called simple tautologies which forms
a simple and large fragment of the set of all \BCK provable formulas.
As we will see afterwards the class of simple tautologies is so big that it can play a role of good approximation of the whole set of \BCK tautologies.
Therefore, quantitative investigations about behavior of the whole set can be nicely approximated by this fragment. See \cite{zai06} for  discussion about quantitative aspects of  simple tautologies.

\bt{defn}{simple} A {\em simple tautology\/} is an implicational formula of the form
$\tau_1 \ar ( \dots \ar ( \tau_p \ar \alpha)\ldots )$
such that $p>0$, $ \alpha$ is a variable and there is at least one component $\tau_i$ identical to
$\alpha$.
\et{defn}

\bt{thm}{simple}
Every simple tautology is \BCK provable.
\et{thm}

\Proof
By $(\textbf{CK})\textbf{K}$ we can prove $\alpha \impl \alpha$.
Using several times axiom (\K) we can add any number of premisses
and prove $\tau_1 \ar ( \dots \ar ( \tau_p \ar (\alpha \ar \alpha))\ldots )$. Using axiom
(\C) we are able to permute premisses to get $\tau_1 \ar ( \dots \ar ( \tau_p \ar \alpha)\ldots )$.
\boksik

\bt{defn}{   }
The smallest class of lambda calculus terms containing {\bf B},
{\bf C} and {\bf I} (resp. {\bf B}, {\bf C} and {\bf K}) and
closed under application and $\beta$-reduction is called the class
of \BCI (resp. \textsf{BCK}) lambda terms.
\et{defn}

\bt{thm}{Hindley}
\begin{enumerate}[(1)]
\item A \BCI lambda term is a closed lambda term $P$ such that
\begin{enumerate}[(i)]
\item for each subterm $\lambda x. M$ of $P$, $x$ occurs free in
$M$ exactly once,
\item each free variable of $P$ has just one free occurrence in $P$.
\end{enumerate}
\item A \BCK lambda term is a closed lambda term $P$ such that
\begin{enumerate}[(i)]
\item for each subterm $\lambda x. M$ of $P$, $x$ occurs free in
$M$ at most once,
\item each free variable of $P$ has just one free occurrence in
$P$.
\end{enumerate}
\end{enumerate}
\et{thm}

\Proof Proof can be found in Roger Hindley's book \cite{hindley}.
\boksik

\bigskip

By Theorem \ref{thm:Hindley} we immediately get that \BCI is a proper subclass of \BCK .

\section{Classes of formulas}

\bt{defn}{langage} The language $\For$ over $k$ propositional
variables $\set{a_1 , \ldots , a_k}$ is defined inductively as:

\begin{align*}
a_i & \in \For \quad \text {for all } i \leq k , \\
\phi \impl \psi & \in \For \quad \text{if } \phi \in \For \text{ and } \psi \in \For .
\end{align*}
\et{defn}

We can now define the usual notation for a formula.
Let $T \in \For$ be a formula. Hence it is of the form
$ A_1 \impl  (A_2 \impl  (\ldots \impl  (A_p  \impl   r(T))) \ldots )$;
we shall write it \[T=A_1,\dots,A_p \impl  r(T).\]
The formulas $A_i$ are called the {\it premisses} of $T$ and the rightmost propositional variable
$r(T)$  of the formula is called the {\it goal} of $T$.
For formula $T$ which is itself a propositional variable obviously $p=0$ and $r(T ) = T$.
To prove quantitative results about \BCK and \BCI logics we need to define several other
classes of formulas, all of them being special
kinds of either tautologies or non-tautologies.

\bt{defn}{classes}
We define the following subsets of $\For$:
\begin{itemize}
\item The set of all \emph{classical tautologies,} $\CL^k$  is the set of formulas which are
$true$ under any $\set{0,1}$ valuation.

\item The set of all \emph{intuitionistic tautologies,} $\INT^k$ is the set of formulas
for which there are closed lambda terms (constructive proofs) of type identical with
the formula.

\item The set of all \emph{Peirce formulas,} $\PEIRCE^k$ is the set of classical tautologies
which are not intuitionistic ones.

\item The set $\BCK^k$ is the set of formulas
for which there are closed lambda \BCK terms of type identical with
the formula.

\item The set  $\BCI^k$ is the set of formulas
for which there are closed lambda \BCI terms of type identical with
the formula.

\item The set of \emph{simple tautologies,} $\G^k$ is the set of expressions that can be written as
\[T=A_1,\dots,A_p\impl  r(T),\]
where at least one of $A_i$'s is the variable $r(T)$.
\\

\item The set of even formulas $\EVEN^k$, is the set of formulas in which each variable occurs even number of times.

\item The set of simple non-tautologies $\SN^k$, is the set of formulas of the form
\[T=A_1,\dots,A_p\impl  r(T),\]
where $\ r(A_i)\neq r(T)$ for all $i$.
\\

\item The set $\LN^k$ is the set of less simple non-tautologies,
defined as the set of formulas of the form
\[T=B_1,\dots, B_{i-1}, C, B_i,\dots, B_p\impl  r(T),\]
such that
\[C= C_1, C_2, \dots, C_q\impl  r(C),\]
where $r(C)=r(T)$, $q \geqslant 1$, and
\[C_1= D_1, D_2, \dots, D_r \impl  r(D),\]
where $r(D) \neq r(T)$, $r\geqslant 0$, and
the following holds: for all $j$,
$r(B_j) \not\in \{ r(T), r(D)\}$ and
$r(D_j) \not\in \{r(T), r(D)\}$.
\end{itemize}
\et{defn}

The obvious relations between classes above are the following. 

\bt{lem}{SN_LN_relations}
$\SN^k \cup \LN^k  \subseteq   \F^k \setminus \CL^k$
\et{lem}

\Proof
Suppose $T=A_1,\dots,A_p \impl  r(T)$ is in $\SN^k$.
Then evaluate,  all of assumptions $A_i$ by $1$ and goal $r(T)$ by $0$ we get that $T \notin \CL$.
Now let $T \in  \LN^k$ and $T$ is in the form described by the definition \ref{defn:classes}.
The shape of $T$ allows us to evaluate $r(T)$ and $r(D)$ by $0$ and all the $r(B_j)$ and $r(D_j)$ by $1$
to see that $T \notin \CL$.
\boksik

\bt{lem}{SN cap_LN}
$\SN^k \cap \LN^k =  \emptyset$
\et{lem}

\Proof
Simply by observing the syntactic structure of both sets.
\boksik

\bigskip


\comment{
\begin{eqnarray*}
  \SN^k \cup \LN^k  & \subset    & \F^k \setminus \CL^k  \label{LN_SN_are_nontautologies} \\
  \SN^k \cap \LN^k  & =          & \emptyset \label{LN_SN_disjoined}\\
  \G^k             & \subsetneq & \BCK ^k  \subsetneq INT^k \subsetneq \;\;CL^k \label{logics_inclusions} \subsetneq \;\; \F^k \setminus (SN^k \cup LN^k)\\
  \PEIRCE^k        & =          &  \CL^k \setminus Int^k \label{peirce_definition}\\
  \BCI^k          & \subsetneq & \EVEN^k \cap \BCK^k \label{inclusions_BCI_BCK}
\end{eqnarray*}
}

\begin{figure}[h]
\begin{center}
\psset{xunit=8pt} \psset{yunit=8pt} \pslinewidth=0.6pt
\begin{pspicture}(0,0)(40,33)
\rput(0,0){\cnode*{0.01}{1}} \rput(0,30){\cnode*{0.01}{2}}
\rput(40,30){\cnode*{0.01}{3}} \rput(40,0){\cnode*{0.01}{4}}
\rput(0,10){\cnode*{0.01}{5}} \rput(30,30){\cnode*{0.01}{6}}
\rput(40,10){\cnode*{0.01}{7}} \rput(30,0){\cnode*{0.01}{8}}
\rput(25,0){\cnode*{0.0}{9}} \rput(25,10){\cnode*{0.01}{10}}
\ncline{1}{2} \ncline{2}{3} \ncline{3}{4} \ncline{4}{1} \ncline{5}{7}
\ncline{6}{8} \ncline{9}{10}
\rput(31,11.5){\cnode*{0.01}{11}} \rput(31,28.5){\cnode*{0.01}{12}}
\rput(39,28.5){\cnode*{0.01}{13}} \rput(39,11.5){\cnode*{0.01}{14}}
\ncline{11}{12} \ncline{12}{13} \ncline{13}{14} \ncline{14}{11}
\rput(31.8,13){\cnode*{0.01}{15}} \rput(31.8,13.75){\cnode*{0.01}{15'}}
\rput(31.8,16.25){\cnode*{0.01}{15''}} \rput(31.8,27){\cnode*{0.01}{16}}
\rput(38.2,27){\cnode*{0.01}{17}} \rput(38.2,13){\cnode*{0.01}{18}}
\ncline{15}{15'} \ncline{15''}{16} \ncline{16}{17} \ncline{17}{18} \ncline{18}{15}
\rput(22,7){\cnode*{0.01}{19}} \rput(34,7){\cnode*{0.01}{20}}
\rput(34,16.5){\cnode*{0.01}{21}} \rput(22,16.5){\cnode*{0.01}{22}}
\ncline{19}{20} \ncline{20}{21} \ncline{21}{22} \ncline{22}{19}
\rput(31.25,16.25){\cnode*{0.01}{23}} \rput(31.25,13.75){\cnode*{0.01}{24}}
\rput(33.75,16.25){\cnode*{0.01}{25}} \rput(33.75,13.75){\cnode*{0.01}{26}}
\ncline{23}{24} \ncline{24}{26} \ncline{25}{26} \ncline{25}{23}
\scriptsize
\rput(15,23){$SN^k: \ Simple \ non-tautologies$}
\rput(10,7){$LN^k: Less \ simple \ non-tautologies$}
\rput(27.5,4.5){$Other$}
\rput(27.5,3){$non$-}
\rput(27.5,1.5){$tautologies$}
\rput(32.5,15){$BCI^k$}
\rput(35,5){$Peirce^k$}
\rput(35,29){$Int^k$}
\rput(35,27.75){$BCK^k$}
\rput(35,25){$G^k$}
\rput(35,24){$Simple$}
\rput(35,23){$tautologies$}
\rput(27.5,14){$EVEN^k$}
\normalsize
\rput(15,31.5){$\For \setminus \CL^k : \ Non-tautologies$}
\rput(35,31.5){$\CL^k : \ Tautologies$}
\end{pspicture}
\caption{ Inclusions summarized in Lemmas \ref{lem:SN_LN_relations} till \ref{lem:BCI_EVEN_BCK} }
\end{center}
\end{figure}
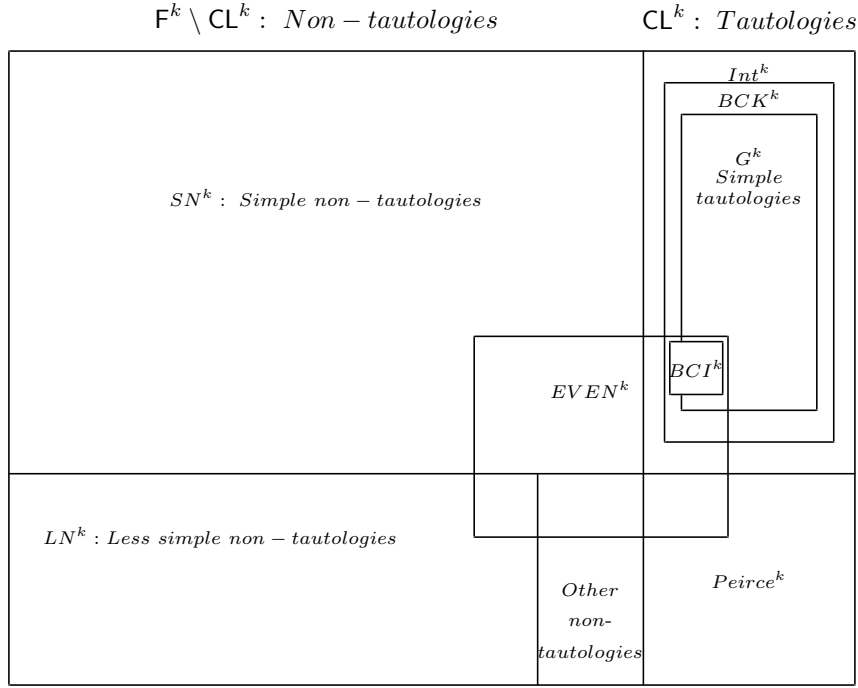

\bigskip

\bt{lem}{G BCK INT CL}
$\G^k \subseteq  \BCK ^k \subseteq INT^k \subsetneq CL^k \subsetneq \F^k \setminus (SN^k \cup LN^k)$
\et{lem}

\Proof
The first inclusion can be found in Theorem \ref{thm:simple}. The rest is trivial.
\boksik

\bt{lem}{BCI_EVEN_BCK}
$\BCI^k \subseteq  \EVEN^k \cap \BCK^k$
\et{lem}

\Proof
$\BCI^k \subsetneq  \EVEN^k$ has been proved recently by Tomasz Kowalski in the paper \cite{kowalski2008} Theorem 6.1. $\BCI^k \subsetneq \BCK^k$ is a classical fact which can be found in \cite{blok_pigozzi89}. See also Theorem \ref{thm:subset}.
\boksik

\section{Densities of sets of formulas}\label{counting}

First we  establish the way in which the size of formula trees are
measured.

\bt{defn}{size} By $ \norm{\phi}$ we mean the size of formula
$\phi$ which we define as total number of leaves in the formula tree $G(\phi)$.
This is in fact the total number of occurrences of propositional variables in the formula.
 Formally,
 \[ \norm{ a_i} =1 \mbox{ and }
 \norm{ \phi \impl \psi } = \norm{\phi} + \norm{\psi} .\]
\end{defn}

\bt{defn}{gestosc} We associate the density $\mu(\X)$ with a
subset $\X \subseteq \For$ of formulas as:

\begin{equation}\label{def}
\mu(\X) = \lim_{n\to\infty} \frac{ \# \{t\in \X: \norm{t}=n \}}{\#
\{t\in \For: \norm{t}=n \}}
\end{equation}

if the limit exists.
 \et{defn}

The number $\mu(\X)$ if it exists is an asymptotic probability of
finding a formula from the class $\X$ among all formulas from
$\For$ or it can be interpreted as the asymptotic density of the
set $\X$ in the set $\For$. It can be immediately seen that the
density $\mu$ is finitely additive so if $\X$ and $\Y$ are
disjoint classes of formulas such that $\mu(\X)$ and $\mu(\Y)$
exist then $\mu(\X \cup \Y)$ also exists and
 $ \mu(\X \cup \Y) = \mu(\X) + \mu(\Y) .$
It is straightforward to observe that for any finite set
$\X$ the density $\mu(\X)$ exists and is $0$. Dually for co-finite
sets $\X$ the density $\mu(\X) = 1$. The density $\mu$ is not
countably additive so in general the formula

\begin{equation} \label{countably_additive}
\mu \left( \bigcup_{i=0}^\infty \X_i \right) =  \sum_{i=0}^\infty
\mu \left( \X_i \right)
\end{equation}

does not hold for all pairwise disjoint classes of sets
$\set{\X_i}_{i\in\mathbb{N}}$. A good counterexample for the
equation (\ref{countably_additive}) is to take as $\X_i$ the singelton
of $i$-th
formula from our language under any natural order of formulas. On
the left hand side of equation \ref{countably_additive} we get
$\mu \left( \For \right)$ which is $1$ but on right hand side
$\mu \left( \X_i \right) = 0$ for all $i\in \mathbb{N}$ and so the
sum is $0$. Finally, we define:

\begin{eqnarray*}
 \mu^- (\X)  &=&  \liminfinfty{\frac{ \#  \{t\in \X: \norm{t}=n \} }{ \# \{t\in \For:
\norm{t}=n \} }}\\
\mu^+ (\X)   &=&  \limsupinfty{ \frac{ \#  \{t\in \X: \norm{t}=n \} }{\#  \{t\in \For:
\norm{t}=n \}}}
\end{eqnarray*}

These two numbers are well defined for any set of formulas  $\X$,
even when the limiting ratio $\mu (\X)$ is not known to exist.

\subsection{Enumerating formulas}\label{counting}

In this section we present some properties of numbers
characterizing the amount of formulas in different classes defined
in our language. Many results and methods
could be rephrased purely in terms of binary trees with given
properties. Obviously an implicational formula from $\For$ of
size $n$ can be seen as a binary tree with $n$ leaves and $k$
labels per leaf (see definitions \ref{formula_tree} and \ref{formula_tree1}).
We will analyze several classes of formulas (trees).

\bt{defn}{F_n^k} By  $\NnF^k_n$ we mean the total number of formulas
from $\For$ of size $n$  so:

\begin{equation}
\NnF^k_n = \# \{ \phi \in \For :\norm{\phi} = n \}.
\end{equation}
 \et{defn}

\bt{lem}{F-numbers} The number $\NnF^k_n = k^n \emph{\texttt{C}}_{n} $ where $\emph{\texttt{C}}_n$ is $(n-1)$th Catalan number.

\comment{
Number $\NnF^k_n$ of all formulas from of the length $n$ is given by the following recursion:
 $\NnF_0^k  = 0, \NnF_1^k =k$ and $F_n^k  =  \sum_{i=1}^{n-1} \NnF_i^k \NnF_{n-i}^k$
}
\et{lem}

\Proof We may use combinatorial observation. A formula from $\For$
of size $n$ can be interpreted as full binary tree of $n$ leaves
with $k$ label per leaf. Therefore for $n=0$ and $n=1$ it is
obvious. Any formula of size $n>1$ is the implication (tree)
between some pair of formulas (trees) of sizes $i$ and $n-i$,
respectively. Therefore the total number of such pairs is
$\sum_{i=1}^{n-1} \NnF_i^k \NnF_{n-i}^k$.
Therefore by simple induction we can immediately
see that $\NnF^k_n = k^n \emph{\texttt{C}}_n $.
For more elaborate treatment of Catalan numbers see Wilf
\cite[pp.~43--44]{Wilf}. We mention only the following well-known
nonrecursive formula for $\emph{\texttt{C}}_n = \frac{1}{n} \Newton{2n-2}{n-1}$.
 \boksik

\bt{lem}{number-of-simple-tautologies} The number $\NnG_n^k$ of
simple tautologies is given by the recursion

\begin{eqnarray}
\NnG_1^k & = & 0, \;\;\;\; \NnG_2^k =  k, \\
\NnG_n^k & = &
\NnF_{n-1}^k - \NnG_{n-1}^k + \sum_{i=2}^{n-1} \NnF_{n-i}^k \NnG_{i}^k.\label{number-of-simple-tautologies}
\end{eqnarray}
\et{lem}

\Proof For the whole discussion about simple tautologies see \cite{zai06}. In particular for a proof look at Lemma 15 of \cite{zai06}.
\boksik

\bt{lem}{number_of_EVEN}
The number $\NnEVEN_n^k$ is given by: \[ \frac{\emph{\texttt{C}}_n}{2^k} \sum_{j=0}^k \Newton{k}{j} (k-2j)^n .\]
\et{lem}

\begin{proof}
The proof is based on the observation obtained in the paper of Franssens (\cite{Fra06} page 30, formula 7.18). In this paper it is obtained an explicit formula
$e_n^k= \frac{1}{2^k} \sum_{j=0}^k \Newton{k}{j} (k-2j)^n$
for the number of closed walks, based at a vertex, of length $n$ along the edges of  $k$-dimensional cube (see also \cite{Stanlay86}). This number in fact appears at Maclaurin series of $\cosh^k (t)$, for all $k$. Note that for odd $n$ this gives $0$.  As we can see the number $e_n^k$ obviously enumerates the set of all sequences of length $n$ of variables $\set{a_1 , \ldots , a_k}$ in which every variable $a_i$ occurs even number of times. Multiplying this by the number $\emph{\texttt{C}}_n$ of all binary trees with $n$ leaves we obtain the explicit formula for $\NnEVEN_n^k$. For some special $k$, sequences $e_n^k$ are mentioned in Sloane's catalogue \cite{sloane}. For instance, $e_{2n}^2 = 2 * 4^{p-2}$ is described in Sloane as A009117.
The sequence $e_{2n}^3 = (3^n+3)/4$ is present in Sloane's catalogue as A054879. Finally $e_{2n}^4$ is Sloane's A092812.
\end{proof}

\subsection{Generating functions}\label{Generating_functions}

 In this paper we investigate the
 proportion between the number of formulas of the size $n$ that are
 tautologies in various logics against the number of all formulas of size $n$ for
 propositional formulas of the language $\For$.
 Our interest lies  in finding limit of that
 fraction when $n$ tends to infinity.
 For this purpose combinatorics has
developed an extremely powerful tool, in the form of generating
series and generating functions.
  A nice exposition of the method
can be found in Wilf \cite{Wilf}, as well as in
in Flajolet, Sedgewick
\cite{fs-book}.
 As the reader may now
expect, while working with formulas we will be often
concerned with complex analysis, analytic functions and their
singularities.

\bigskip

 Let $A= (A_0,A_1,A_2,\dots  )$ be a sequence of real numbers. The
{\em ordinary generating series\/} for $A$ is the formal power
series $\sum_{n=0}^\infty A_nz^n.$
 And, of course, formal power
series are in one-to-one correspondence to sequences. However,
considering $z$ as a complex variable, this series, as known from
the theory of analytic functions, converges uniformly to a
function $f_{A}(z)$ in some open disc $\{z \in \Complex : \abs{z} < R\}$ of
maximal diameter, and $R \geq 0$ is called its radius of
convergence. So with the sequence $A$ we can associate a complex
function $f_A(z),$ called the {\em ordinary generating function\/}
for $A,$ defined in a neighborhood of $0.$
 This correspondence is
one-to-one again (unless $R=0$), since, as it is well known from
the theory of analytic functions, the expansion of a complex
function $f(z),$ analytic in a neighborhood of $z_0,$ into a power
series $\sum_{n=0}^\infty A_n(z-z_0)^n$ is unique.

 Many questions concerning the asymptotic behavior of $A$ can be
efficiently resolved by analyzing the behavior of its generating function $f_{A}$ at the
complex circle $|z|=R.$
 This is the approach we take to determine the asymptotic fraction
of tautologies and many other classes of formulas among all
formulas of a given size.

\bigskip

The main tool used to obtain limits of the fraction of two sequences which are described by
generating functions  will be the following result, due to Szeg\"o
\cite{Sz} [Thm.\ 8.4], see as well Wilf \cite{Wilf} [Thm.\ 5.3.2 page 181].
We can see versions of Szeg\"o lemma in action in papers \cite{zaionc05}, \cite{zai06}
concerning asymptotic probabilities in logic.
The second powerful tool we need is so called Drmota-Lalley-Woods theorem which has been developed independently by Drmota in \cite{Drm97},
Lalley in \cite{Lall93} and Woods in \cite{Woods97} to study problems involving the enumeration of families of plane trees or context-free languages and finding their asymptotic behaviour as solutions of positive algebraic systems. The best presentation of Drmota-Lalley-Woods theorem can be found in
Flajolet and Sedgewick in \cite{fs-book} [pp. 446-451].
Excellent overview of the metod is due to Daniele Gardy in \cite{gardy-dmtcs} [Chapter 4, pp 15-16].

\bigskip

Theorems and lemmas in the next chapter \ref{densities_of_classes} are proved using Szeg\"o lemma. The result mentioned in Theorems \ref{thm:classical} and \ref{thm:BCK/classical} that the limiting ratio $\mu(\CL^k)$ of classical tautologies with $k$ propositional variables exists, requires the use of
Drmota-Lalley-Woods theorem.

\subsection{Densities of classes of formulas}\label{densities_of_classes}

In this section we wish to summarize results
contained in the  papers
\cite{zaionc05} and \cite{FGGZ07}.

\bt{lem}{klasa_G_k}

The asymptotic probability of the fact that a randomly chosen  formula is a simple tautology is:

\begin{equation}\label{limGform}
\mu(\G^k) = \lim_{n\to\infty} \frac{\NnG^k_n}{\NnF^k_n}=\frac{4k+1}{(2k+1)^2}
\nonumber
\end{equation}
\et{lem}

\begin{proof}
The first proof of this fact can be found in \cite{mtz00}. A simpler one is in Theorem 30 of  \cite{zaionc05} page 252.
\end{proof}

\bt{lem}{prop_densite_SN}
The density of simple non-tautologies exists and is equal to
\[ \mu(\SN^k) = \lim_{n\to\infty} \frac{\NnSN^k_n}{\NnF^k_n} = \frac{k(k-1)}{(k+1)^2}.\]
For large $k$, this density is $1-3/k+\Theta(1/k^2).$
\et{lem}

\begin{proof}
This result was already given in the paper \cite[page 586]{mtz00}. The alternative proof can be found in \cite{FGGZ07} at Proposition 7.
\end{proof}

\bt{lem}{prop_densite_LN}
The density of less simple non-tautologies is equal to
\[ \mu(\LN^k) = \lim_{n\to\infty} \frac{\NnLN^k_n}{\NnF^k_n} = \frac{2k(k-1)^2}{(k+2)^4}.\]
For large $k$ it is equal to $2/k + \Theta(1/k^2)$.
\et{lem}

\begin{proof}
The long and complicated proof of this fact can be found in  chapter 4.3 of \cite{FGGZ07}.
\end{proof}

\bt{thm}{classical}
Asymptotically (for a large number $k$ of Boolean variables),
all classical tautologies are simple and it follows that  all intuitionistic tautologies are classical i.e.

\begin{eqnarray*}
 \limkinfty \frac{\mu (\G^k)}{\mu(\CL^k)}   &=&  1,\\
 \limkinfty \frac{\mu^- (\Int^k)}{\mu(\CL^k)} &=&  1.
\end{eqnarray*}
\et{thm}

\begin{proof}
We know that for any $k$, the limiting ratio $\mu(Cl^k)$
of classical tautologies with $k$ propositional
variables exists. This result is obtained by standard
techniques in analysis of algorithms; we skip the details and refer the interested reader
to Flajolet and Sedgewick \cite{fs-book} or to Gardy \cite{gardy-dmtcs}.
From the fact presented in Lemma \ref{lem:G BCK INT CL}

\[\G^k \subseteq \INT^k \subseteq \CL^k \subseteq \F^k \setminus (\SN^k \cup \LN^k)\]

it follows

\[ \mu (G^k) \leq \mu^- (\INT^k ) \leq  \mu (\CL^k) \leq  1 - \mu (\SN^k) - \mu (\LN^k) .
\]
Now our result follows since lower and upper bounds
 $\mu(\G^k)$ and $1 - \mu (\SN^k) - \mu (\LN^k)$ are equal to $1/k + \Theta(1/k^2)$.
\end{proof}

\bt{lem}{Peirces}
The density $\mu(\PEIRCE^k)$ of Peirce formulas (if it exists) is equal to $\frac{1}{2k^2}$.
\[ \liminf_{n\to\infty} \frac{\NnPEIRCE^k_n}{F^k_n} = \limsup_{n\to\infty} \frac{\NnPEIRCE^k_n}{F^k_n} = \frac{1}{2 k^2}.\]
\et{lem}

\begin{proof}
Proof of this fact appears  in the paper \cite{GKM}.
\end{proof}

\bt{thm}{EVAN}
For the set $\EVEN^k$ of formulas we have:
\begin{eqnarray*}
\mu^-(\EVEN^k ) & = & \liminfinfty{\frac{\NnEVEN^k_n}{F^k_n} }= 0 \\
\mu^+ (\EVEN^k) & = & \limsupinfty{ \frac{\NnEVEN^k_n}{F^k_n}} = \frac{1}{2^{k-1}}. \\
\end{eqnarray*}
\et{thm}

\begin{proof}
The first equality is trivial since
$\NnEVEN^k_{n} =
 \frac{\emph{\texttt{C}}_n}{2^k} \sum_{j=0}^k \Newton{k}{j} (k-2j)^n$ is zero  by lemma
\ref{lem:number_of_EVEN} for all odd numbers $n$. Now suppose that $n=2m$.

\begin{eqnarray*}
\limsup_{2m \to \infty} \frac{\NnEVEN^k_{2m}}{F^k_{2m}} & = &
\lim_{2m\to\infty} \frac{\emph{\texttt{C}}_{2m}}{2^k} \sum_{j=0}^k \Newton{k}{j} \frac{(k-2j)^{2m}}{F^k_{2m}} \\
                               &  = & \frac{1}{2^k} \sum_{j=0}^k \lim_{2m\to\infty} \Newton{k}{j} \frac{\emph{\texttt{C}}_{2m} (k-2j)^{2m}}{k^{2m} \emph{\texttt{C}}_{2m}} \\
                               & = & \frac{1}{2^k} \lim_{2m\to\infty} \Newton{k}{0} \frac{\emph{\texttt{C}}_{2m} k^{2m}}{k^{2m} \emph{\texttt{C}}_{2m}} +
\frac{1}{2^k} \lim_{2m\to\infty} \Newton{k}{k} \frac{\emph{\texttt{C}}_{2m} (- k)^{2m}}{k^{2m} \emph{\texttt{C}}_{2m}} \\
                              & = & \frac{1}{2^k} + \frac{1}{2^k} = \frac{1}{2^{k-1}},\\
\end{eqnarray*}

where the second and the third line in the above display are equal as
$$\lim_{2m\to\infty}
\Newton{k}{j} \frac{\emph{\texttt{C}}_{2m} (k-2j)^{2m}}{k^{2m} \emph{\texttt{C}}_{2m}} = 0$$
for every $0<j<k$.
\end{proof}

The picture below summarizes all theorems involving densities
which are needed for proving
our next two results on \BCK and \BCI logics namely
Theorems
\ref{thm:BCK/classical} and \ref{thm:BCI/BCK}.

\begin{figure}[h]
\begin{center}
\psset{xunit=8pt} \psset{yunit=8pt} \pslinewidth=0.6pt
\begin{pspicture}(0,0)(40,33)
\rput(0,0){\cnode*{0.01}{1}} \rput(0,30){\cnode*{0.01}{2}}
\rput(40,30){\cnode*{0.01}{3}} \rput(40,0){\cnode*{0.01}{4}}
\rput(0,10){\cnode*{0.01}{5}} \rput(30,30){\cnode*{0.01}{6}}
\rput(40,10){\cnode*{0.01}{7}} \rput(30,0){\cnode*{0.01}{8}}
\rput(25,0){\cnode*{0.01}{9}} \rput(25,10){\cnode*{0.01}{10}}
\ncline{1}{2} \ncline{2}{3} \ncline{3}{4} \ncline{4}{1} \ncline{5}{7}
\ncline{6}{8} \ncline{9}{10}
\rput(31,11.5){\cnode*{0.01}{11}} \rput(31,28.5){\cnode*{0.01}{12}}
\rput(39,28.5){\cnode*{0.01}{13}} \rput(39,11.5){\cnode*{0.01}{14}}
\ncline{11}{12} \ncline{12}{13} \ncline{13}{14} \ncline{14}{11}
\rput(31.8,13){\cnode*{0.01}{15}} \rput(31.8,13.75){\cnode*{0.01}{15'}}
\rput(31.8,16.25){\cnode*{0.01}{15''}} \rput(31.8,27){\cnode*{0.01}{16}}
\rput(38.2,27){\cnode*{0.01}{17}} \rput(38.2,13){\cnode*{0.01}{18}}
\ncline{15}{15'} \ncline{15''}{16} \ncline{16}{17} \ncline{17}{18} \ncline{18}{15}
\rput(22,7){\cnode*{0.01}{19}} \rput(34,7){\cnode*{0.01}{20}}
\rput(34,16.5){\cnode*{0.01}{21}} \rput(22,16.5){\cnode*{0.01}{22}}
\ncline{19}{20} \ncline{20}{21} \ncline{21}{22} \ncline{22}{19}
\rput(31.25,16.25){\cnode*{0.01}{23}} \rput(31.25,13.75){\cnode*{0.01}{24}}
\rput(33.75,16.25){\cnode*{0.01}{25}} \rput(33.75,13.75){\cnode*{0.01}{26}}
\ncline{23}{24} \ncline{24}{26} \ncline{25}{26} \ncline{25}{23}
\scriptsize
\rput(15,23){$SN^k: \ Simple \ non-tautologies$}
\rput(15,20){$\frac{k(k-1)}{(k+1)^2}=1-\frac3k + \Theta\left(\frac{1}{k^2}\right)$}
\rput(10,7){$LN^k: Less \ simple \ non-tautologies$}
\rput(10,3){$\frac{2k(k-1)^2}{(k+2)^4}=\frac2k + \Theta\left(\frac{1}{k^2}\right)$}
\rput(27.5,4.5){$Other$}
\rput(27.5,3){$non$-}
\rput(27.5,1.5){$tautologies$}
\rput(32.5,15){$BCI^k$}
\rput(35,5){$Peirce^k$}
\rput(35,29){$Int^k$}
\rput(35,27.75){$BCK^k$}
\rput(35,25){$G^k$}
\rput(35,24){$Simple$}
\rput(35,23){$tautologies$}
\rput(35,21){$\frac{4k+1}{(2k+1)^2}$}
\rput(35,19){$=\frac1k + \Theta\left( \frac{1}{k^2}\right)$}
\rput(27.5,14){$EVEN^k$}
\rput(27,12){$\frac{1}{2^{k-1}}$}
\normalsize
\rput(15,31.5){$\mathcal{F}^k \setminus Cl^k : \ Non-tautologies$}
\rput(35,31.5){$Cl^k : \ Tautologies$}
\end{pspicture}
\caption{Densities of the sets of formulas}
\end{center}
\end{figure}
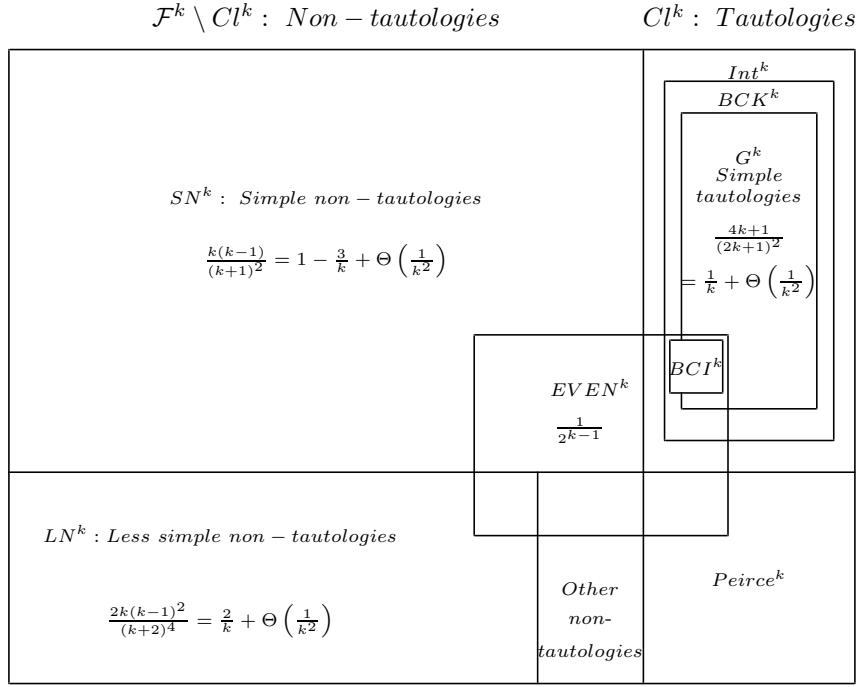

\bt{thm}{BCK/classical}
Almost every classical tautology is \BCK provable i.e.
$$ \limkinfty \frac{\mu^- (\BCK^k)}{\mu(\CL^k)} = 1.$$
\et{thm}

\begin{proof}
From the fact that $\G^k \subseteq  \BCK^k \subseteq  \Int^k \subseteq \CL^k$, we have
\[ \mu (G^k) =
\limkinfty \frac{\NnG^k_n}{\NnF^k_n} \leqslant
\liminfinfty \frac{\NnBCK^k_n}{\NnF^k_n}
\leqslant \limsupinfty\frac{\NnInt^k_n}{\NnF^k_n}
\leqslant \liminfty \frac{\NnCL^k_n}{\NnF^k_n} = \mu (\CL^k).
\]
The result follows from the fact that
both $\mu(\G^k)$ and $\mu(\CL^k)$ are equal to $1/k + \Theta(1/k^2)$.
\end{proof}

\bt{thm}{BCI/BCK}
Almost non \BCK provable formula is \BCI provable  i.e.
$$ \limkinfty \frac{\mu^+ (\BCI^k)}{\mu^- (\BCK^k)} = 0.$$
\et{thm}

\begin{proof}
It follows from two facts. Lemma \ref{lem:BCI_EVEN_BCK} shows  $\BCI^k \subseteq \EVEN^k$  while from Lemma  \ref{lem:G BCK INT CL} based on Theorem \ref{thm:simple}
we get $\G^k \subset \BCK^k$.
The rest is based on the calculations from Theorem \ref{thm:EVAN} and Lemma \ref{lem:klasa_G_k}.

 Therefore:
\[ \mu^+ (\BCI^k) \leq \mu^+ (\EVEN^k) \leq \frac{1}{2^{k-1}},\]
while
\[
\mu^- (\BCK^k) \geq \mu (\G^k ) = \frac{4k+1}{(2k+1)^2} .
\]
Finally we have

\[ \limkinfty \frac{\mu^+ (\BCI^k)}{\mu^- (\BCK^k)} \leq \limkinfty \frac{(2k+1)^2 }{2^{k-1} (4k+1) } =  0.\]

\end{proof}

\textbf{Interpretation of result:} Weakening rule $\varphi \impl (\psi \impl \varphi) $
is much stronger tool to generate formulas  then identity rule $\varphi \impl \varphi$.

\bigskip

\textbf{Open problems:} We do not know if there exist densities
$\mu (\BCI^k)$ or $\mu (\BCK^k)$
 of two investigated logics


\section{Counting proofs in \BCI and \BCK logics}

In this section we will focus on two special classes \BCI and \BCK of lambda
terms. On the basis of their special structure,
we will show how to enumerate \BCI and \BCK terms of a given size. As
the main result the density of \BCI terms among \BCK terms is computed.


\bt{defn}{   }
The size of a lambda term is defined in the following way:
\begin{align*}
\| x\| &= 1\\
\| \lambda x.M \| &= 1 + \| M\| \\
\| MN \| &= 1+ \| M\| + \| N\|.
\end{align*}
\et{defn}
As we can see $\|t \| $ is the number of all  nodes of lambda tree $G(t)$.

\begin{defn}
Let $n$ be an integer. We denote by $\Lambda_n$ the set of all closed lambda terms
up to $\alpha$ conversion
of size $n$. Obviously the set $\Lambda_n$ is finite. We denote its cardinality by $L_n$.
\end{defn}

As far as we know, no asymptotic analysis of the sequence $L_n$ has been done. Moreover, typical combinatorial techniques do not seem to apply easily for this task.
For the first time the problem of enumerating lambda terms was
considered in \cite{wang}. In general, counting lambda
terms of a given size turns out to be a non-trivial and challenging task. The wide
discussion on this problem and some results concerning properties
of random terms can be found in \cite{properties}.

\subsection{Enumerating \BCI terms}

\bt{defn}{ }
By $a_n$ we denote the number of \BCI terms of size $n$.
\et{defn}

Since in \BCI terms each lambda binds exactly one variable, in a lambda
tree for such a term the number of leaves is equal to
the number of unary nodes. In every unary-binary tree the number
of leaves is greater by one than the number of binary nodes. Thus
the number of \BCI terms is positive only if the size is equal to $3k+2$ ($k$
binary nodes, $k+1$ unary nodes and $k+1$ leaves) for $k \in
{\mathbb N}$.

\bt{defn}{ }
By $a^*_n$ we denote the number of \BCI terms up to $\alpha$ conversion
with $n$ binary nodes.
\et{defn}

Obviously, $a^*_n=a_{3n+2}$. Moreover $a_n =0$ for $n \neq 2$ mod $3$. 

\begin{lem}
The sequence $(a^*_n)$ satisfies the recurrence:
\begin{align*}
&a^*_0=1, \qquad a^*_1=5,\\
&a^*_n=6na^*_{n-1}+\sum_{i=1}^{n-2}a^*_i a^*_{n-i-1} \quad , \text{ for } n \geq 2.
\end{align*}
\end{lem}

\begin{proof}
There is only one \BCI term of size $2$ (no binary nodes): $\lambda x.x$. Moreover there
are five terms of size $5$ (one binary node): $\lambda xy.xy$, $\lambda xy.yx$, $(\lambda x.x)(\lambda x.x)$, $\lambda x.(\lambda y.y)x$ and $\lambda x.x(\lambda y.y)$. Thus, $a^*_0=1$ and $a^*_1=5$.

Let $P$ be a \BCI term with $n \geq 2$ binary nodes. Such
a term is either in the form of application or in the form of
abstraction. Both cases are depicted in Figure 4. 

In the first case $P$ is an application of two \BCI terms, $P
\equiv MN$, where $M$ has $i$ binary nodes and $N$ has $n-i-1$
binary nodes ($i=0, \ldots , n-1$). It gives us $\sum_{i=0}^{n-1}
a^*_i a^*_{n-i-1}$ possibilities.

In the second case $P$ is in the form of abstraction, $P \equiv
\lambda x. M$, and $x$ occurs free in $M$ exactly once. In the tree
corresponding to $M$, the parent of the leaf labeled with $x$
must be a binary node. Thus, we can look at the tree corresponding
to $M$ as at the lambda tree for some \BCI term $Q$ with $n-1$ binary
nodes with an additional leaf labeled with $x$. 
This leaf can be inserted
into the tree in two manners, either on the left or on the right.
Moreover this can be done in $3n-1$ ways which is the number of all
 branches in the tree for $Q$. 
Thus, there are $2(3n-1)a^*_{n-1}$ possibilities of such insertions.

\begin{figure}[h]
\begin{center}
\psset{xunit=25pt} \psset{yunit=25pt} \pslinewidth=0.5pt
\begin{pspicture}(1,0)(13.5,4)
\rput(1,0){\cnode*{0.01}{1}} \rput(3,0){\cnode*{0.01}{2}}
\rput(2,2){\cnode*{0.01}{3}} \rput(4,0){\cnode*{0.01}{4}}
\rput(6,0){\cnode*{0.01}{5}} \rput(5,2){\cnode*{0.01}{6}}
\rput(3.5,3){\cnode*{0.03}{7}}
\ncline{1}{2} \ncline{2}{3} \ncline{3}{1} \ncline{4}{5} \ncline{5}{6} \ncline{6}{4}
\ncline[linestyle=dotted]{3}{7} \ncline[linestyle=dotted]{6}{7}
\rput(3.5,3.25){$@$}
\rput(2,.5){$a^*_i$}
\rput(5,.5){$a^*_{n-i-1}$}
\rput(8,0){\cnode*{0.01}{11}} \rput(10,0){\cnode*{0.01}{12}}
\rput(9,2){\cnode*{0.01}{13}} \rput(9,3.5){\cnode*{0.03}{14}}
\ncline{11}{12} \ncline{12}{13} \ncline{13}{11}
\ncline[linestyle=dotted]{13}{14}
\rput(9,.5){$a^*_{n-1}$} \rput(9,3.75){$\lambda x$}
\rput(10.5,1){$+$} \rput(12,0.8){or}
\rput(11.7,1.5){$@$} \rput(12.3,1.5){$@$} \rput(11.2,0.3){$x$} \rput(12.8,0.3){$x$}
\rput(11.7,1.3){\cnode*{0.03}{15}} \rput(11.2,0.5){\cnode*{0.03}{16}}
\rput(12.3,1.3){\cnode*{0.03}{17}} \rput(12.8,0.5){\cnode*{0.03}{18}}
\ncline{15}{16} \ncline{17}{18}
\end{pspicture}
\caption{Two ways of obtaining a \BCI term with $n \geq 2$ binary nodes}
\end{center}
\end{figure}
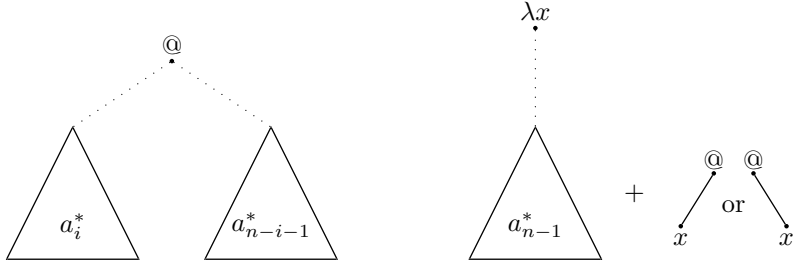

Summing up we get the equation of the lemma. 

\end{proof}

Denoting by $A(x)$ the generating function for the sequence $(a^*_n)$ and after basic calculations, we get
$$6x^2 \frac{\partial A(x)}{\partial x} + xA^2(x) + (4x-1)A(x) + 1 =0, \qquad A(0)=1.$$
This is a non-linear Riccati differential equation and as such it has a solution which is a non-elementary function.

The sequence $(a^*_n)$ and the function $A(x)$ were studied in \cite{Janson}. On the basis of that paper we get the asymptotics
$$a_{3n+2} = a^*_n \sim \frac{1}{2 \pi}6^n(n-1)!.$$

First values of $(a_n)$ are the following:
$$0, 0, 1, 0, 0, 5, 0, 0, 60, 0, 0, 1105, 0, 0, 27120, 0, 0, 828250, 0, 0, 30220800,  \ldots$$
This sequence can be found in the On-Line Encyclopedia of Integer Sequences (\cite{sloane})
under the number  A062980.

\subsection{Enumerating \BCK terms}

\bt{defn}{ }
Let us denote by $b_n$ the number of \BCK lambda terms of size $n$.
\et{defn}

\begin{lem}
The sequence $(b_n)$ satisfies the following recursive equation:
\begin{align*}
&b_0=b_1=0, \qquad b_2=1, \qquad b_3=2, \qquad b_4=3,\\
&b_n = b_{n-1} + 2\sum_{i=0}^{n-3} i b_i + \sum_{i=0}^{n-1} b_i b_{n-i-1} +1 \quad \text{ for } n\geq 5.
\end{align*}
\end{lem}

\begin{proof}
There are no terms of size $0$ and $1$, there is only one \BCK term of size $2$: $\lambda x.x$, two terms of size $3$: $\lambda xy.x$ and $\lambda xy.y$, and three terms of size $4$: $\lambda xyz.x$, $\lambda xyz.y$ and $\lambda xyz.z$. Thus $b_0=b_1=0$, $b_2=1$, $b_3=2$ and $b_4=3$.

Let $P$ be a \BCK term of size $n \geq 5$. Such a term is either in the form of application or in the form of
abstraction where the first lambda binds one variable or in the form of abstraction where the first lambda does not bind any variables.

In the first case $P$ is in the form of application, $P \equiv MN$, where $M$ and $N$ are \BCK terms of size, respectively, $i$ and $n-i-1$ ($i=0,\ldots,n-1$). It gives us $\sum_{i=0}^{n-1} b_i b_{n-i-1}$ possibilities.

In the second case $P$ is in the form of abstraction, $P \equiv \lambda x.M$ and $x$ occurs free in $M$ exactly once.
There are two subcases here: either $M\equiv \lambda x_1 \ldots x_{n-2} . x$ or $M$ is a term built of a \BCK term $Q$ of size $i=2,\ldots,n-3$ with an additional term $\lambda x_1 \ldots x_{n-i-3}. x$ inserted on one of its branches or on the branch joining $\lambda x$ with $Q$. Since in $Q$ there are $i-1$ branches and the additional term can be inserted either on the left or on the right, this case gives us $1+2\sum_{i=2}^{n-3} i b_i$ possibilities. Both subcases are presented in Figure 5.

\begin{figure}[h]
\begin{center}
\psset{xunit=25pt} \psset{yunit=25pt} \pslinewidth=0.5pt
\begin{pspicture}(0,0)(11,4)
\rput(0,0){\cnode*{0.03}{1}} \rput(0,0.75){\cnode*{0.03}{2}}
\rput(0,2.25){\cnode*{0.03}{3}} \rput(0,3){\cnode*{0.03}{4}}
\ncline{1}{2} \ncline[linestyle=dotted]{2}{3} \ncline{3}{4}
\rput[Bl](0.3,0){$x$} \rput[l](0.3,0.75){$\lambda x_{n-2}$} \rput[l](0.3,2.25){$\lambda x_{1}$} \rput[l](0.3,3){$\lambda x$}
\rput(4,0){\cnode*{0.01}{11}} \rput(6,0){\cnode*{0.01}{12}}
\rput(5,2){\cnode*{0.01}{13}} \rput(5,3.5){\cnode*{0.03}{14}}
\ncline{11}{12} \ncline{12}{13} \ncline{13}{11}
\ncline[linestyle=dotted]{13}{14}
\rput(5,.5){$b_i$} \rput(5,3.75){$\lambda x$}
\rput(6.5,1){$+$} \rput(9.5,0.8){or}
\scriptsize
\rput(8.8,1.9){$@$} \rput[r](8.2,1.3){$\lambda x_1$} \rput[r](8.2,0.6){$\lambda x_{n-i-3}$} \rput[r](8.2,0.1){$x$}
\rput(8.8,1.7){\cnode*{0.03}{15}} \rput(8.4,1.3){\cnode*{0.03}{16}}
\rput(8.4,0.6){\cnode*{0.03}{17}} \rput(8.4,0.1){\cnode*{0.03}{18}}
\ncline{15}{16} \ncline[linestyle=dotted]{16}{17} \ncline{17}{18}
\rput(10.2,1.9){$@$} \rput[l](10.8,1.3){$\lambda x_1$} \rput[l](10.8,0.6){$\lambda x_{n-i-3}$} \rput[l](10.8,0.1){$x$}
\rput(10.2,1.7){\cnode*{0.03}{19}} \rput(10.6,1.3){\cnode*{0.03}{20}}
\rput(10.6,0.6){\cnode*{0.03}{21}} \rput(10.6,0.1){\cnode*{0.03}{22}}
\ncline{19}{20} \ncline[linestyle=dotted]{20}{21} \ncline{21}{22}
\end{pspicture}
\caption{Second case of the construction of a \BCK term of size $n \geq 5$}
\end{center}
\end{figure}
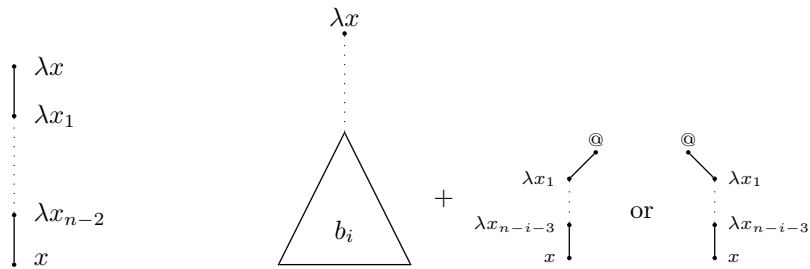

In the third case $P$ is in the form of abstraction, $P \equiv \lambda x.M$ and $x$ does not occur free in $M$. The number of such terms of size $n$ is equal to the number of all \BCI terms of size $n-1$. Thus, it gives us $b_{n-1}$ possibilities.

Summing up we get the equations of the lemma. 
\end{proof}

Denoting by $B(x)$ the generating function for the sequence $(b_n)$ and after basic calculations, we get

$$2x^4 \frac{\partial B(x)}{\partial x} + (x-x^2) B^2(x) - (1-x)^2 B(x) + x^2= 0, \qquad B(0)=0.$$

Again we obtained a non-linear Riccati differential equation. Unfortunately, this time we do not know the solution.

First values of $(b_n)$ are the following:

$$0, 0, 1, 2, 3, 9, 30, 81, 225, 702, 2187, 6561, 19602, 59049, 177633, 532170, 1594323, \ldots$$

Also this sequence can be found in \cite{sloane} under the number A073950.

\subsection{Density of \BCI in \BCK}
Our main goal is to compute the density of \BCI in \BCK . Since there are no \BCI terms of size $n \not\equiv 2 \mod 3$ (which is not true in the case of \BCK terms), if the density exists, then it is $0$.

Let us observe that each \BCK term can be obtained from a \BCI term with some additional (possibly none) lambdas. This observation allows us to obtain a formula for $b_n$ depending on the sequence $(a_n)$.

\begin{lem}\label{combrep}
The number of possible ways of choosing $k$ out of the $n$ elements with repetition is equal to ${n+k-1 \choose n-1}$.
\end{lem}

\begin{lem}\label{ogr}
For $k \in {\mathbb N}$ the following formula holds
$$b_{3k+2} = \sum_{i=0}^{k} {3k \choose 3i} a_{3i+2}.$$
\end{lem}

\begin{proof}
Each \BCK term of size $3k+2$ can be obtained from a \BCI term of size $3i+2$ ($i=0,\ldots,k$) by inserting $3k+2-3i-2=3k-3i$ additional lambdas. There are $3i+1$ branches in the tree for a \BCI term of size $3i+2$, thus the number of insertions corresponds to the number of $3k-3i$-element combinations with repetition of a $3i+1$-element set and thus, by Lemma \ref{combrep}, it is equal to ${3k \choose 3k-3i} = {3k \choose 3i}$.
\end{proof}

\bt{thm}{  }
The density of \BCI terms among \BCK terms equals 0.
\et{thm}

\begin{proof}
By the asymptotics of $(a_{3k+2})$ and by Lemma \ref{ogr}, we get
\[ \frac{a_{3k+2}}{b_{3k+2}} \leq \frac{a_{3k+2}}{a_{3k+2} + {3k \choose 3} a_{3k-1}} \overset{k\to \infty}{\longrightarrow} 0. \qedhere \]
\end{proof}


\begin{thebibliography}{}
\bibitem{blok_pigozzi89} Blok, W.J., Pigozzi, D. Algebraizable Logics, \emph{Memoirs of
the A.M.S.}, 77, Nr. 396, 1989.

\bibitem{Curry_Hindley_Seldin} H.B. Curry, R. Hindley and J.P. Seldin, Combinatory Logic II, North Holland, 1972.

\bibitem{CFGG04}
B.~Chauvin, P.~Flajolet, D.~Gardy and B.~Gittenberger.
And/Or trees revisited.
{\em Combinatorics, Probability and Computing},
13(4-5): 475--497, 2004.

\bibitem{properties} Ren{\'e} David, Katarzyna Grygiel, Jakub Kozik, Christophe Raffalli, Guillaume Theyssier, Marek Zaionc, {\it Some properties of random $\lambda$-terms}, draft published at Computing Research Repository (CoRR), 2009, arXiv:0903.5505v1.

\bibitem{Drm97} M. Drmota. Systems of functional equations,
 \emph{Random Structures and Algorithms, } 10, pp103-124, 1997.



\bibitem{fs-book}
P.~Flajolet and R.~Sedgewick.
\newblock {\em Analytic Combinatorics}.
\newblock Cambridge University Press, 2009.

\bibitem{FGGZ07} H. Fournier, D. Gardy, A. Genitrini and M. Zaionc.
\emph{Classical and intuitionistic logic are asymptotically identical}.
In Proc. \emph{Internat. Conf. on Computer Science Logic (CSL)},
LNCS Vol. 4646: 177--193, 2007.

\bibitem{Fra06} Ghislain R. Franssens, \emph{On a Number Pyramid Related to the Binomial, Deleham, Eulerian, MacMahon and Stirling number triangles},  Journal of Integer Sequences, Volume 9, 2006,
    Article 06.4.1 pp 1-34.

\bibitem{gardy-dmtcs} D. Gardy.
Random Boolean expressions.
{\em Colloquium on Computational Logic and Applications},
Chamb\'ery (France), June 2005.
Proceedings in \emph{Discrete Mathematics and Theoretical Computer Science}, AF: 1--36, 2006.

\bibitem{GW05} D. Gardy and A. Woods.
And/Or tree probabilities of Boolean function.
First Colloquium on the Analysis of Algorithms, Barcelona, 2005.
Proceedings in \emph{Discrete Mathematics and Theoretical Computer Science}, AD: 139--146, 2005.

\bibitem{GK} A. Genitrini and J. Kozik.
Quantitative comparison of Intuitionistic and Classical logics -- full propositional system.
\emph{International Symposium on Logical Foundations of Computer Science}, Florida (USA),
Jan. 2009. Proceedings LNCS Vol.~5407: 280--294.

\bibitem{GKM} A. Genitrini, J. Kozik and G. Matecki.
\newblock On the density and the structure of the Peirce-like formulae.
\newblock In Proc. {\em Fifth Colloquium on Mathematics and Computer Science},
Blaubeuren, Germany, September 2008. DMTCS Proceedings.


\bibitem{hindley} J. Roger Hindley, {\it Basic Simple Type Theory}. Cambridge Tracts in Theoretical Computer Science,
Vol. 42, Cambridge: Cambridge University Press, 1997.

\bibitem{Janson} Svante Janson, {\it The Wiener index of simply generated random trees}. Random Structures Algorithms 22 (2003) 337-358.

\bibitem{kos-zaionc03} Z. Kostrzycka and M. Zaionc.
Statistics of intuitionistic versus classical logic.
{\em Studia Logica}, 76(3): 307--328, 2004.

\bibitem{kowalski2008} T. Kowalski, \emph{Tomasz Self-implications in BCI}. Notre Dame Journal of Formal Logic 49 (2008), no. 3, 295--305

\bibitem{kozik08}
J.~Kozik.
\newblock Subcritical pattern languages for {A}nd/{O}r trees.
\newblock In Proc. {\em Fifth Colloquium on Mathematics and Computer Science},
Blaubeuren, Germany, september 2008. DMTCS Proceedings.

\bibitem{Lall93}  S.P. Lalley. Finite range random walk on free groups and homogeneous trees,
 \emph{Ann. Probab.,} 21(4), pp 2087-2130, 1993.

\bibitem{LS97} H. Lefmann and P. Savick\'y.
Some typical properties of large {And/Or} {B}oolean expressions.
{\em Random Structures and Algorithms}, 10: 337--351, 1997.

\bibitem{mat05} G. Matecki.
Asymptotic density for equivalence.
{\em Electronic Notes in Theoretical Computer Science}, 140: 81--91, 2005.

\bibitem{mtz00} M. Moczurad, J. Tyszkiewicz and M. Zaionc.
Statistical properties of simple types.
{\em Mathematical Structures in Computer Science},
10(5): 575--594, 2000.

\bibitem{sloane} Neil J. A. Sloane. The on-line encyclopedia of integer sequences. Available at: \mbox{\url{http://www.research.att.com/~njas/sequences/index.html}}


\bibitem{Stanlay86} R. P. Stanley, \emph{Enumerative Combinatorics},
Wadsworth and Brooks/Cole, 1986.


\bibitem{Sz} G. Szeg\"o, Orthogonal polynomials,
         fourth ed., AMS, Colloquium Publications, 23,
         Providence, 1975.

\bibitem{wang} Jue Wang, {\it Generating Random Lambda Calculus Terms}. Available at: \mbox{\url{http://cspeople.bu.edu/juewang/research.html}}

\bibitem{Wilf} H. Wilf.
\emph{Generatingfunctionology}.
Second edition, Academic Press, Boston, 1994.

\bibitem{w05} A. Woods.
On the probability of absolute truth for And/Or formulas.
\emph{The Bulletin of Symbolic Logic}, 12(3): 523, 2006.

\bibitem{Woods97}  A. Woods. Coloring rules for finite trees, and probabilities of monadic second order sentences, \emph{Random Structures and Algorithms,} 10, pp453-485, 1997.

\bibitem{wron83}  Wro\'{n}ski, A. BCK-algebras do Not Form a Variety, \emph{Math. Japonica}, 28, pp. 211-213, 1983.

\bibitem{zaionc05} M. Zaionc.
On the asymptotic density of tautologies in logic of implication and negation.
\emph{Reports on Mathematical Logic}, 39: 67--87, 2005.

\bibitem{zai06} M. Zaionc.
Probability distribution for simple tautologies.
\emph{Theoretical Computer Science}, 355(2): 243--260, 2006.

\end{thebibliography}
\end{document}